\documentclass[12pt]{amsart}
\usepackage{amsthm}
\usepackage{amsmath}
\usepackage{amssymb}
\usepackage{amscd}

\newtheorem{theorem}[subsection]{Theorem}
\newtheorem{lemma}[subsection]{Lemma}
\newtheorem{proposition}[subsection]{Proposition}
\newtheorem{corollary}[subsection]{Corollary}

\theoremstyle{definition}

\theoremstyle{remark}

\newtheorem{example}[subsection]{Example}

\newcommand{\Id}{\operatorname{Id}}

\newcommand{\GL}{\operatorname{GL}}

\newcommand{\gl}{\operatorname{gl}}

\vspace{1.5cm}
\title{{\bf The strict AKNS hierarchy:\\its structure and solutions}}

\author{G.F. Helminck,\\
Korteweg-de Vries Institute\\
University of Amsterdam\\
P.O. Box 94248\\
1090 GE Amsterdam\\
The Netherlands\\
e-mail: g.f.helminck@uva.nl}


\begin{document}
\maketitle

\begin{abstract}
In this paper we discuss an integrable hierarchy of compatible Lax equations that is obtained by a wider deformation of a commutative algebra in the loop space of ${\rm sl}_{2}$ than that in the AKNS-case
and whose Lax equations are based on a different decomposition of this loop space.  We show the compatibility of these Lax equations and that they are equivalent to a set of zero curvature relations. We present a linearization of the system and conclude by giving a wide construction of solutions of this hierarchy.
\end{abstract}

\section{\boldmath Introduction} 
\label{intro}

Integrable hierarchies often occur as the evolution equations of the generators of a deformation of a commutative subalgebra inside some Lie algebra $\mathfrak{g}$. The deformation and the evolution equations are determined by a splitting of $\mathfrak{g}$
in the direct sum of two Lie subalgebras, like in the Adler-Kostant-Symes Theorem, see \cite{A-vM-VH}. This gives then rise to a compatible set of Lax equations, a so-called hierarchy, and the simplest nontrivial equation often determines the name of the hierarchy. The hierarchy we discuss here corresponds to a somewhat different splitting of an algebra of ${\rm sl}_2$ loops that w.r.t. the usual decomposition yields the AKNS-hierarchy. Following the terminology used in similar situations, see \cite{HHP}, we use the name {\it strict} AKNS-hierarchy for the system of Lax equations corresponding to this different decomposition. 

The letters AKNS refer to the following fact: Ablowitz, Kaup, Newell and Segur showed, see \cite{AKNS}, that the initial value
problem of the following system of equations for two complex functions $q$ and $r$ depending of the variables $x$ and $t$:
\begin{align}
\label{akns1}
&i\dfrac{\partial}{\partial t}q(x,t) :=  iq_t =-\dfrac{1}{2} q_{xx} + q^2r\\ \notag
&i\dfrac{\partial}{\partial t}r(x,t) :=ir_t  =\dfrac{1}{2}r_{xx}-qr^2.
\end{align}
could be solved with the Inverse
Scattering Transform and that is why this system (\ref{akns1}) is called the AKNS-equations. It can be rewritten in a zero curvature form for $2 \times 2$-matrices as follows:
consider the matrices
$$
P_{1}=\left(
\begin{matrix}
0& q\\
r& 0
\end{matrix}
\right) \text{ and }P_{2}=\left(
\begin{matrix}
-\frac{i}{2}qr & \frac{i}{2}q_x \\
-\frac{i}{2}r_x & \frac{i}{2}qr
\end{matrix}
\right).
$$
A direct computation shows that there hold for them the following identities
$$
\dfrac{\partial}{\partial x}(P_2)=\left(
\begin{matrix}
\frac{i}{2}(q_{x}r+qr_{x})& \frac{i}{2} q_{xx}\\
-\frac{i}{2} r_{xx}& -\frac{i}{2}(q_{x}r+qr_{x})
\end{matrix}
\right) 
$$
and
$$
[P_{2},P_{1}]=\left(
\begin{matrix}
\frac{i}{2}(q_{x}r+qr_{x})& iq^{2}r\\
-iqr^{2}& -\frac{i}{2}(q_{x}r+qr_{x})
\end{matrix}
\right).
$$
By combining them, one sees that the AKNS-equations are equivalent to
\begin{equation}
\label{}
\label{akns2}
\dfrac{\partial}{\partial t}(P_1)-\dfrac{\partial}{\partial x}(P_2) 
-[P_2,P_1]=0
\end{equation}
In the next section we will show that equation (\ref{akns2}) is the simplest compatibility condition of a system of Lax equations that consequently carries the name AKNS-hierarchy, see \cite{FNR}. Besides that we introduce here also another system of Lax equations for a more general deformation of the basic generator of the commutative algebra. It relates to a different decomposition of the relevant Lie algebra. Since its analogue for the KP hierarchy is called the strict KP hierarchy, we use in the present context the term strict AKNS hierarchy. We show first of all that this new system is compatible and equivalent with a set of zero curvature relations. Further we describe suitable linearizations that can be used to construct solutions of both hierarchies. 
We conclude by constructing solutions of both hierarchies.

\section{The AKNS-hierarchy and its strict version}
\label{SAKNS}

We present here an algebraic description of the AKNS-hierarchy and its strict version that underlines the deformation character of these hierarchies as pointed at in the introduction. At the equation (\ref{akns2}) one has to deal with 2 $\times$ 2-matrices whose coefficients were polynomial expressions of a number of complex functions and their derivatives w.r.t. some parameters. The same is true for the other equations in both hierarchies. We formalize this as follows: our starting point is a commutative complex algebra $R$ that should be seen as the source from which the coefficients of the 2 $\times$ 2-matrices are taken. We will work in the Lie algebra ${\rm sl}_{2}(R)[z, z^{-1})$ consisting of all elements 
\begin{equation}
\label{lsl2}
X=\sum_{i=-\infty}^{N} X_{i}z^{i}, \text{ with all }X_{i} \in {\rm sl}_{2}(R)
\end{equation}
 and the bracket 
 $$
 [X,Y]=[\sum_{i=-\infty}^{N} X_{i}z^{i} , \sum_{j=-\infty}^{M} Y_{j}z^{j} ]:=\sum_{i=-\infty}^{N} \sum_{j=-\infty}^{M} [X_{i}, Y_{j}] z^{i+j}.
 $$
We will also make use of the slightly more general Lie algebra ${\rm gl}_{2}(R)[z, z^{-1})$, where the coefficients in the $z$-series from (\ref{lsl2}) are taken from ${\rm gl}_{2}(R)$ instead of ${\rm sl}_{2}(R)$ and the bracket is given by the same formula. In the Lie algebra ${\rm gl}_{2}(R)[z, z^{-1})$ we decompose elements in two ways. The first is as follows:
\begin{equation}
\label{eltdeco}
X=\sum_{i=-\infty}^{N} X_{i}z^{i}=\sum_{i=0}^{N} X_{i}z^{i} +\sum_{i=-\infty}^{-1} X_{i}z^{i}=:(X)_{\geqslant 0}+(X)_{<0}
\end{equation}
and this induces the splitting
\begin{equation}
\label{AKNSdeco}
{\rm gl}_{2}(R)[z, z^{-1})={\rm gl}_{2}(R)[z, z^{-1})_{\geqslant 0} \oplus {\rm gl}_{2}(R)[z, z^{-1})_{<0}, 
\end{equation}
where the two Lie subalgebras ${\rm gl}_{2}(R)[z, z^{-1})_{\geqslant 0}$ and ${\rm gl}_{2}(R)[z, z^{-1})_{<0}$ are given by
\begin{align*}
&{\rm gl}_{2}(R)[z, z^{-1})_{\geqslant 0}=\{ X \in {\rm gl}_{2}(R)[z, z^{-1}) \mid X=(X)_{\geqslant 0}=\sum_{i=0}^{N} X_{i}z^{i} \} \text{ and }\\
&{\rm gl}_{2}(R)[z, z^{-1})_{<0}=\{ X \in {\rm gl}_{2}(R)[z, z^{-1}) \mid X=(X)_{<0}=\sum_{i <0} X_{i}z^{i} \}.
\end{align*}
By restriction it leads to a similar decomposition for ${\rm sl}_{2}(R)[z, z^{-1})$, which is relevant for the AKNS hierarchy. 
The second way to decompose elements of 
${\rm gl}_{2}(R)[z, z^{-1})$ is: 
\begin{equation}
\label{eltdeco2}
X=\sum_{i=-\infty}^{N} X_{i}z^{i}=\sum_{i=1}^{N} X_{i}z^{i} +\sum_{i=-\infty}^{0} X_{i}z^{i}=:(X)_{>0}+(X)_{\leqslant 0}.
\end{equation}
This yields the splitting
\begin{equation}
\label{SAKNSdeco}
{\rm gl}_{2}(R)[z, z^{-1})={\rm gl}_{2}(R)[z, z^{-1})_{>0} \oplus {\rm gl}_{2}(R)[z, z^{-1})_{\leqslant 0}. 
\end{equation}
By restricting it to ${\rm sl}_{2}(R)[z, z^{-1})$, we get a similar decomposition for this Lie algebra, which relates, as we will see further on, to the strict version of the AKNS hierarchy. The Lie subalgebras ${\rm gl}_{2}(R)[z, z^{-1})_{> 0}$ and ${\rm gl}_{2}(R)[z, z^{-1})_{\leqslant 0}$ in (\ref{SAKNSdeco}) are defined in a similar way as the first two Lie subalgberas
\begin{align*}
&{\rm gl}_{2}(R)[z, z^{-1})_{> 0}=\{ X \in {\rm gl}_{2}(R)[z, z^{-1}) \mid X=(X)_{> 0}=\sum_{i=1}^{N} X_{i}z^{i} \} \text{ and }\\
&{\rm gl}_{2}(R)[z, z^{-1})_{\leqslant 0}=\{ X \in {\rm gl}_{2}(R)[z, z^{-1}) \mid X=(X)_{\leqslant 0}=\sum_{i \leqslant 0} X_{i}z^{i} \}.
\end{align*}

Both hierarchies correspond to evolution equations for deformations of the generators of a commutative Lie algebra in the first component.
Denote the matrix  $\left(
\begin{smallmatrix}
-i& 0\\
0&i
\end{smallmatrix}
\right)$ by $Q_{0}$. Inside ${\rm sl}_{2}(R)[z, z^{-1})_{\geqslant 0}$ this commutative complex Lie subalgebra $C_{0}$ is chosen to be the Lie subalgebra with basis the $\{Q_{0}z^{m} \mid m \geqslant 0 \}$ and in ${\rm sl}_{2}(R)[z, z^{-1})_{> 0}$ our choice will be the Lie subalgebra $C_{1}$ with the basis $\{Q_{0}z^{m} \mid m \geqslant 1 \}$.
In the first case that we work with $C_{0}$, we assume that the algebra $R$ possesses a set $\{ \partial_{m} \mid m \geqslant 0\}$ of commuting $\mathbb{C}$-linear derivations $\partial_{m}: R \to R$, where each $\partial_{m}$ should be seen as the derivation corresponding to the flow generated by $Q_{0}z^{m}$. The data $(R, \{ \partial_{m} \mid m \geqslant 0\})$ is also called a {\it setting} for the AKNS-hierarchy. In the case that we work with the decomposition (\ref{SAKNSdeco}), we merely need a set of commuting $\mathbb{C}$-linear  derivations $\partial_{m}: R \to R, m\geqslant 1,$ with the same interpretation. Staying in the same line of terminology, we call the data $(R, \{ \partial_{m} \mid m \geqslant 1 \})$ a {\it setting} for the strict AKNS-hierarchy.
\begin{example}
\label{E2.1}
Examples of settings for the respective hierarchies are the algebras of complex polynomials $\mathbb{C}[t_{m}]$ in the variables $\{ t_{m} \mid m \geqslant 0\}$ resp. $\{ t_{m} \mid m \geqslant 1\}$ or the formal power series $\mathbb{C}[[t_{m}]]$ in the same variables, with both algebras equipped with the derivations $\partial_{m}=\frac{\partial}{\partial t_{m}}, m \geqslant 0$, in the AKNS-set-up or $\partial_{m}=\frac{\partial}{\partial t_{m}}, m \geqslant 1$, in the strict AKNS-case.
\end{example}
We let each derivation $\partial_{m}$ occurring in some setting act coefficient wise on the matrices from ${\rm gl}_{2}(R)$ and that defines then a derivation of this algebra. The same holds for the extension to ${\rm gl}_{2}(R)[z, z^{-1})$ defined by 
$$
\partial_{m}(X):=\sum_{j=-\infty}^{N} \partial_{m}(X_{j})z^{j}.
$$

We have now sufficient ingredients to discuss the AKNS-hierarchy and its strict version. In the AKNS-case our interest is in certain deformations of the Lie algebra $C_{0}$ obtained by conjugating with elements of the group corresponding to ${\rm gl}_{2}(R)[z, z^{-1})_{<0}$. At the strict version we are interested
in certain deformations of the Lie algebra $C_{1}$ obtained by conjugating with elements from a group linked to ${\rm gl}_{2}(R)[z, z^{-1})_{\leqslant 0}$.

Note that for each $X \in {\rm gl}_{2}(R)[z, z^{-1})_{<0}$ the exponential map yields a well-defined element of the form
\begin{equation}
\label{exp}
\exp(X)=\sum_{k=0}^{\infty} \frac{1}{k!}X^{k}=\Id +Y, Y \in {\rm gl}_{2}(R)[z, z^{-1})_{<0}
\end{equation}
and with the formula for the logarithm one retreives $X$ back from $Y$. One verifies directly that the elements of the form (\ref{exp}) form a group w.r.t. multiplication and this
we see as the group $G_{<0}$ corresponding to ${\rm gl}_{2}(R)[z, z^{-1})_{<0}$. 

 In the case of the Lie subalgebra ${\rm gl}_{2}(R)[z, z^{-1})_{\leqslant 0}$, one cannot move back and forth between the Lie algebra and its group. Nevertheless, one can assign a proper group to this Lie algebra. A priori, the exponential $\exp(Y)$ of an element $Y \in {\rm gl}_{2}(R)[z, z^{-1})_{\leqslant 0}$, does not have to define an element in ${\rm gl}_{2}(R)[z, z^{-1})_{\leqslant 0}$. That requires convergence conditions. However, if it does, then it belongs to  
$$
G_{\leqslant 0}=\{ K=\sum_{j=0}^{\infty}K_{j}z^{-j} \mid \text{ all } K_{j}\in {\rm gl}_{2}(R), K_{0} \in {\rm gl}_{2}(R)^{*} \},
$$
where ${\rm gl}_{2}(R)^{*}$ denotes the elements in ${\rm gl}_{2}(R)$ that have a multiplicative inverse in ${\rm gl}_{2}(R)$. It is a direct verification that $G_{\leqslant 0}$ is a group and we see it as a proper group corresponding to the Lie algebra ${\rm gl}_{2}(R)[z, z^{-1})_{\leqslant 0}$. In fact, $G_{\leqslant 0}$ is isomorphic to the semi-direct 
product of $G_{<0}$ and ${\rm gl}_{2}(R)^{*}$.

Now there holds
\begin{lemma} 
\label{L2.1}
The group $G_{\leqslant 0}$ acts by conjugation on ${\rm sl}_{2}(R)[z, z^{-1})$.
\end{lemma}
\begin{proof}
Take first any $g \in G_{<0}$. Then there is an $X \in {\rm gl}_{2}(R)[z, z^{-1})_{<0}$ such that $g=\exp(X)$. Since there holds for every $Y \in {\rm sl}_{2}(R)[z, z^{-1})$
$$
gYg^{-1}=\exp(X)Y \exp(-X)=e^{\text{ad}(X)}(Y)=Y+ \sum_{k=1}^{\infty} \frac{1}{k!} \text{ad}(X)^{k}(Y)
$$
and this shows that the coefficients for the different powers of $z$ in this expression are commutators of elements of ${\rm gl}_{2}(R)$ and ${\rm sl}_{2}(R)$ and that proofs the claim for elements from $G_{<0}$. Since conjugation with an element from ${\rm gl}_{2}(R)^{*}$ maps ${\rm sl}_{2}(R)$ to itself, the same holds for ${\rm sl}_{2}(R)[z, z^{-1})$. This proves the full claim.
\end{proof}
Next we have a look at the different deformations.
The deformations of $C_{0}$ by elements of $G_{<0}$ are determined by that of $Q_{0}$. Therefore we focus on that element and we consider for some $g=\exp(X)=\exp(\sum_{j=1}^{\infty}X_{j}z^{-j}) \in G_{<0}$ the deformation 
\begin{align}
\label{decoQ0}
\notag
Q&=gQ_{0}g^{-1}:=\sum_{i=0}^{\infty}Q_{i}z^{-i}\\
&=\left(
\begin{matrix}
-i& 0\\
0&i
\end{matrix}
\right) + [X_{1},Q_{0}]z^{-1} +\{ [X_{2},Q_{0}] + \frac{1}{2} [X_{1}, [X_{1}, Q_{0}]]\}z^{-2} +\cdots 
\end{align}
of $Q_{0}$. From this formula we see directly that, if each $X_{i}=\left(
\begin{matrix}
-\alpha_{i}& \beta_{i}\\
\gamma_{i}&\alpha_{i}
\end{matrix}
\right), i=1,2,$ then
\begin{equation}
\label{coefQ1}
Q_{1}=\left(
\begin{matrix}
0& 2i \beta_{1}\\
-2i \gamma_{i}&0
\end{matrix}
\right)=\left(
\begin{matrix}
0&q \\
r&0
\end{matrix}
\right),
\end{equation}
and
\begin{equation}
\label{}
Q_{2}=\left(
\begin{matrix}
q_{11}& q_{12}\\
q_{21}&q_{22}
\end{matrix}
\right)=\left(
\begin{matrix}
-2i \beta_{1} \gamma_{1}& 2i(\beta_{2}-\alpha_{1}\beta_{1}) \\
-2i(\gamma_{2}+\alpha_{1}\gamma_{1})&2i \beta_{1} \gamma_{1}
\end{matrix}
\right)
\end{equation}
In particular we get in this way that $q_{11}=-i\frac{qr}{2}$ and $q_{22}=i\frac{qr}{2}$. 
Since $Q$ is the deformation of $Q_{0}$ and $\Id z$ is central, the deformation of each $Q_{0}z^{m}, m \in \mathbb{Z} ,$, is given by $Qz^{m}$. The deformation of the Lie algebra $C_{1}$ by elements from $G_{\leqslant 0}$ are basically determined by that of the element $Q_{0}z$. So we focus on the deformation of this element. Using the same notations as at the deformation of $Q_{0}$ by $G_{<0}$, we get that the deformation of $Q_{0} z$ by a $Kg \in G_{\leqslant 0}$, with $K \in {\rm gl}_{2}(R)^{*}$ and $g \in G_{<0}$, looks like
\begin{align}
\label{sdecoQ0}
\notag
Z&=KgQ_{0}zg^{-1}K^{-1}:=\sum_{i=0}^{\infty}KQ_{i}K^{-1}z^{1-i}= \sum_{i=0}^{\infty}Z_{i}z^{1-i}\\
&=
Z_{0}z+ [KX_{1}K^{-1},Z_{0}] +
\cdots .
\end{align}
Consequently, the corresponding deformation of each $Q_{0} z^{m}, m \geqslant 1,$ is $Zz^{m-1}$.

Now we are looking for deformations $Q$ of the form (\ref{decoQ0}) such that the evolution w.r.t. the $\{ \partial_{m} \}$ satisfies: for all $m \geqslant 0$ 
\begin{equation}
\label{LaxAKNS}
\partial_{m}(Q)=[(Qz^{m})_{ \geqslant 0},Q]=-[(Qz^{m})_{ < 0},Q],
\end{equation}
where the second identity follows from the fact that all $\{ Qz^{m} \}$ commute.
The equations (\ref{LaxAKNS}) are called the {\it Lax equations of the AKNS hierarchy} and the deformation $Q$ satisfying these equations is called a {\it solution} of the hierarchy. Note that $Q=Q_{0}$ is a solution of the AKNS hierarchy and it is called the {\it trivial} one.
The AKNS equation (\ref{akns2}) occurs among the compatibility equations of this system, as we will see further on. Note that the equation (\ref{LaxAKNS}) for $m=0$ is simply $\partial_{0}(Q)=[Q_{0}, Q]$. Therefore, if $\partial_{0}=\frac{\partial}{\partial t_{0}}$ and both $e^{-it_{0}}$ and $e^{it_{0}}$ belong to the algebra $R$ of matrix coefficients, then we can introduce the loop $\hat{Q} \in {\rm sl}_{2}(R)[z, z^{-1})$ given by
$$
\hat{Q}:=\exp (-t_{0}Q_{0})Q \exp (t_{0}Q_{0})=\sum_{k \geqslant  0}\exp (-t_{0}Q_{0})Q_{k} \exp (t_{0}Q_{0})z^{-k},
$$
which is easily seen to satisfy $\partial_{0}(\hat{Q})=0$. This handles then the dependence of $Q$ of $t_{0}$.

Among the deformations $Z$ of the form (\ref{sdecoQ0}) we look for $Z$ such that their evolution w.r.t. the $\{ \partial_{m} \}$ satisfies: for all $m \geqslant 1$ 
\begin{equation}
\label{LaxSAKNS}
\partial_{m}(Z)=[(Zz^{m-1})_{ > 0},Z]=-[(Zz^{m-1})_{ \leqslant 0},Z],
\end{equation}
where the second identity follows from the fact that all $\{ Z z^{m-1} \}$ commute.
Since the equations (\ref{LaxSAKNS}) correspond to the strict cut-off (\ref{eltdeco2}), they are called the {\it Lax equations of the strict AKNS hierarchy} and the deformation $Z$ is called a {\it solution} of this hierarchy. Again there is al least one solution $Z=Q_{0}z$. It is called the {\it trivial} solution of the hierarchy.

For both systems (\ref{LaxAKNS}) and (\ref{LaxSAKNS}) one can speak of compatibility. There holds namely
\begin{proposition}
\label{P2.1}
Both sets of Lax equations (\ref{LaxAKNS}) and (\ref{LaxSAKNS}) are so-called compatible systems, i.e. the projections $\{ B_{m}:=(Qz^{m})_{\geqslant 0} \mid m \geqslant 0 \}$ of a solution $Q$ to (\ref{LaxAKNS}) satisfy the zero curvature relations
\begin{equation}
\label{ZCAKNS}
\partial_{m_{1}}(B_{m_{2}})-\partial_{m_{2}}(B_{m_{1}})-[B_{m_{1}},B_{m_{2}}]=0
\end{equation}
and the projections $\{ C_{m}:=(Zz^{m-1})_{> 0} \mid m \geqslant 1 \}$ corresponding to a solution $Z$ to (\ref{LaxSAKNS}) satisfy the zero curvature relations
\begin{equation}
\label{ZCSAKNS}
\partial_{m_{1}}(C_{m_{2}})-\partial_{m_{2}}(C_{m_{1}})-[C_{m_{1}},C_{m_{2}}]=0
\end{equation}
\end{proposition}
\begin{proof}
The idea is to show that the left hand side of (\ref{ZCAKNS}) resp. (\ref{ZCSAKNS}) belongs to 
$$
{\rm sl}_{2}(R)[z, z^{-1})_{\geqslant 0} \cap {\rm sl}_{2}(R)[z, z^{-1})_{<0} \text{ resp. }
{\rm sl}_{2}(R)[z, z^{-1})_{> 0} \cap {\rm sl}_{2}(R)[z, z^{-1})_{\leqslant 0}
$$ 
and thus has to be zero. We give the proof for the $\{ C_{m} \}$, that for the $\{ B_{m} \}$ is similar and is left to the reader. The inclusion in the first factor is clear as both $C_{m}$ and $\partial_{n}(C_{m})$ belong to the Lie subalgebra ${\rm sl}_{2}(R)[z, z^{-1})_{ > 0}.$ 
To show the other one, we use the Lax equations (\ref{LaxSAKNS}). 
Note that the same Lax equations hold for all the $\{ z^{N}Z \mid N \geqslant 0 \}$
\begin{equation*}
\partial_{m}(z^{N}Z)=[(Zz^{m-1})_{ > 0},z^{N}Z].
\end{equation*}
By substituting $C_{m_{i}}=z^{m_{i}-1}Z-(z^{m_{i}-1}Z)_{<0}$ we get for
\begin{align*}
\partial_{m_{1}}(C_{m_{2}})-\partial_{m_{2}}(C_{m_{1}})=&\;\partial_{m_{1}}(z^{m_{2}-1}Z)-\partial_{m_{1}}((z^{m_{2}-1}Z)_{\leqslant 0})\\
 &\; -\partial_{m_{2}}(z^{m_{1}-1}Z)+\partial_{m_{2}}((z^{m_{1}-1}Z)_{\leqslant 0})\\
=&\; [C_{m_{1}},z^{m_{2}-1}Z]-[C_{m_{2}},z^{m_{1}-1}Z]\\
& \;-\partial_{m_{1}}((z^{m_{2}-1}Z)_{ \leqslant 0})+\partial_{m_{2}}((z^{m_{1}-1}Z)_{\leqslant 0})
\end{align*} 
and for 
\begin{align*}
[C_{m_{1}},C_{m_{2}}]=&\;[z^{m_{1}-1}Z-(z^{m_{1}-1}Z)_{\leqslant 0}, z^{m_{2}-1}Z-(z^{m_{2}-1}Z)_{\leqslant 0}]\\
=&\;-[(z^{m_{1}-1}Z)_{\leqslant 0}, z^{m_{2}-1}Z]+[(z^{m_{2}-1}Z)_{\leqslant 0}, z^{m_{1}-1}Z]\\
&\;+[(z^{m_{1}-1}Z)_{\leqslant 0},(z^{m_{2}-1}Z)_{\leqslant 0}].
\end{align*}
Taking into account the second identity in (\ref{LaxSAKNS}), we see that the left hand side of (\ref{ZCSAKNS}) is equal to 
$$
-\partial_{m_{1}}((z^{m_{2}-1}Z)_{\leqslant 0})+\partial_{m_{2}}((z^{m_{1}-1}Z)_{\leqslant 0})-[(z^{m_{1}-1}Z)_{\leqslant 0},(z^{m_{2}-1}Z)_{\leqslant 0}].
$$
This element belongs to the Lie subalgebra ${\rm sl}_{2}(R)[z, z^{-1})_{\leqslant 0}$ and that proves the claim.
\end{proof}
Reversely, we have
\begin{proposition} 
\label{P2.2}
Suppose we have a deformation $Q$ of the type (\ref{decoQ0}) or a deformation $Z$ of the form (\ref{sdecoQ0}). Then there holds:
\begin{enumerate}
\item Assume that the projections $\{ B_{m}:=(Qz^{m})_{\geqslant 0} \mid m \geqslant 0 \}$ satisfy the zero curvature relations (\ref{ZCAKNS}). Then $Q$ is a solution of the AKNS-hierarchy.
\item Similarly, if the projections $\{ C_{m}:=(Zz^{m-1})_{\geqslant 0} \mid m \geqslant 1 \}$ satisfy the zero curvature relations (\ref{ZCSAKNS}), then $Z$ is a solution of the strict AKNS-hierarchy.
\end{enumerate}
\end{proposition}
 \begin{proof}
 Again we prove the statement for $Z$, that for $Q$ is shown in a similar way.
 So, assume that there is one Lax equation (\ref{LaxSAKNS}) that does not hold. Then there is a $m \geqslant 1$ such that 
$$
\partial_{m}(Z)-[C_{m},Z]=\sum_{j \leqslant k(m)} X_{j}z^{j}, \text{ with }X_{k(m)} \neq 0.
$$
Since both $\partial_{m}(Z)$ and $[C_{m},Z]$ are of order smaller than or equal to one in $z$, we know that $k(m) \leqslant 1$. Further, we can say that for all 
$N\geqslant 0 $
$$
\partial_{m}(z^{N}Z)-[C_{m},z^{N}Z]=\sum_{j \leqslant k(m)} X_{j}z^{j+N}, \text{ with }X_{k(m)} \neq 0
$$
and we see by letting $N$ go to infinity that the right hand side can obtain any sufficiently large order in $z$. By the zero curvature relation for $N$ and $m$ we get for the left hand side
\begin{align*}
\partial_{m}(z^{N}Z)-[C_{m},z^{N}Z]&=\partial_{m}(C_{N}) -[C_{m}, C_{N}]+\partial_{m}((z^{N}Z)_{\leqslant 0}) -[C_{m}, (z^{N}Z)_{\leqslant0}]\\
&=\partial_{N}(C_{m})+ \partial_{m}((z^{N}Z)_{ \leqslant 0}) -[C_{m}, (z^{N}Z)_{\leqslant 0}]
\end{align*}
and this last expression is of order smaller or equal to $m$ in $z$. This contradicts the unlimited growth in orders of $z$ of the right hand side. Hence all Lax equations (\ref{LaxSAKNS}) have to hold for $Z$.
\end{proof}
Because of the equivalence between the Lax equations (\ref{LaxAKNS}) for $Q$ and the zero curvature relations (\ref{ZCAKNS}) for the $\{ B_{m} \}$, we call this last set of equations also the {\it zero curvature form} of the AKNS-hierarchy. Similarly, the zero curvature relations (\ref{ZCSAKNS}) for the $\{ C_{m} \}$ is called the {\it zero curvature form} of the strict AKNS-hierarchy.

Besides the zero curvature relations for the cut-off's $\{ B_{m}\}$ resp. $\{ C_{m}\}$ corresponding to respectively a solution $Q$ of the AKNS-hierachy and a solution $Z$ of the strict AKNS-hierachy, also other parts satisfy such relations. Define
$$
A_{m}:=B_{m}-Qz^{m}, m \geqslant 0, \text{ and }D_{m}:=C_{m}-Zz^{m-1}, m\geqslant 1,
$$
Then we can say

\begin{corollary}The following relations hold:
\label{C1.1}
\begin{itemize}
\item The parts $\{ A_{m} \mid m \geqslant 0\}$ of a solution $Q$ of the AKNS-hierarchy satisfy 
$$
\partial_{m_{1}}(A_{m_{2}})-\partial_{m_{2}}(A_{m_{1}})-[A_{m_{1}},A_{m_{2}}]=0.
$$
\item
The parts $\{ D_{m} \mid m \geqslant 0\}$ of a solution $Z$ of the strict AKNS-hierarchy satisfy 
$$
\partial_{m_{1}}(D_{m_{2}})-\partial_{m_{2}}(D_{m_{1}})-[D_{m_{1}},D_{m_{2}}]=0
$$
\end{itemize}
\end{corollary}
\begin{proof}
Again we show the result only in the strict case. Note that the $\{Zz^{m-1}\}$ satisfy Lax equations similar to $Z$
$$
\partial_{i}(Zz^{m-1})=[D_{i}, Zz^{m-1}], i \geqslant 1.
$$
Now we substitute in the zero curvature relations for the $\{ C_{m}\}$ everywhere the relation $C_{m}=D_{m}+Zz^{m-1}$ and use the above Lax equations and the fact that all the $\{Zz^{m-1}\}$ commute. This gives the desired result.
\end{proof}

To clarify the link with the AKNS-equation, consider the relation (\ref{ZCAKNS}) for $m_{1}=2, m_{2}=1$ 
$$
\partial_{2}(Q_{0}z +Q_{1})=\partial_{1}(Q_{0}z^{2} +Q_{1}z +Q_{2})+[Q_{2}, Q_{0}z +Q_{1}].
$$
Then this identity reduces in ${\rm sl}_{2}(R)[z, z^{-1})_{ \geqslant 0}$, since $Q_{0}$ is constant, to the following two equalities
\begin{align}
\label{ZC12}
\partial_{1}(Q_{1})=[Q_{0}, Q_{2}] \text{ and }\partial_{2}(Q_{1}) =\partial_{1}(Q_{2})+[Q_{2},Q_{1}].
\end{align}
The first gives an expression of the off-diagonal terms of $Q_{2}$ in the coefficients $q$ and $r$ of $Q_{1}$, i.e.
$$
q_{12}=\frac{i}{2}\partial_{1}(q) \text{ and }q_{21}=-\frac{i}{2}\partial_{1}(r)
$$ 
and the second equation becomes the AKNS-equations (\ref{akns2}), if one has $\partial_{1}=\frac{\partial}{\partial x}$ and $\partial_{2}=\frac{\partial}{\partial t}$.

\section{The linearization of both hierarchies}
\label{linhier}

The zero curvature form of both hierarchies points at the possible existence of a linear system of which the zero curvature equations form the compatibility conditions. We present here such a system for each hierarchy. For the AKNS hierarchy, this system, the {\it linearization of the AKNS hierarchy}, is as follows: take a potential solution 
$Q$ of the form (\ref{decoQ0}) and consider the system
\begin{equation}
\label{LinAKNS}
Q \psi= \psi Q_{0}, \partial_{m}(\psi)=B_{m}\psi, \text{ for all $m \geqslant 0$ and $B_{m}=(Qz^{m})_{\geqslant 0}$}.
\end{equation}
Likewise, for a potential solution $Z$ of the strict AKNS hierarchy of the form (\ref{sdecoQ0}), the {\it linearization of the strict AKNS hierarchy} is given by
\begin{equation}
\label{LinSAKNS}
Z \varphi= \varphi Q_{0}z, \partial_{m}(\varphi)=C_{m}\varphi, \text{ for all $m \geqslant 1$ and $C_{m}=(Zz^{m-1})_{> 0}$}.
\end{equation}
Before specifying $\psi$ and $\varphi$, we show the manipulations needed to derive the Lax equations of the hierarchy from its linearization. We do this for the strict AKNS hierarchy, for the AKNS hierarchy one proceeds similarly. Apply $\partial_{m}$ to the first equation in (\ref{LinSAKNS}) and apply both equations (\ref{LinSAKNS}) in the sequel:
\begin{align}
\notag
\partial_{m}(Z\varphi -\varphi Q_{0}z)&=\partial_{m}(Z )\varphi +Z\partial_{m}(\varphi)-\partial_{m}(\varphi)Q_{0}z=0\\ \notag
&=\partial_{m}(Z )\varphi+ZC_{m} \varphi-C_{m} \varphi Q_{0}z\\ \label{LinLax}
&=\left\{ \partial_{m}(Z )-[C_{m},Z] \right\} \varphi=0.
\end{align}
Hence, if we can scratch $\varphi$ from the equation (\ref{LinLax}), then we obtain the Lax equations of the strict AKNS hierarchy. To get the proper set-up where all these manipulations make sense, we first have a look at the linearization for the trivial solutions $Q=Q_{0}$ and $Z=Q_{0}z$. In the AKNS-case we have then
\begin{equation}
\label{tsolAKNS}
Q_{0} \psi_{0}= \psi_{0} Q_{0}, \partial_{m}(\psi)=Q_{0}z^{m}\psi_{0}, \text{ for all $m \geqslant 0$}.
\end{equation}
and for its strict version
\begin{equation}
\label{tsolSAKNS}
Q_{0}z \varphi_{0}= \varphi_{0} Q_{0}z, \partial_{m}(\varphi_{0})=Q_{0}z^{m} \varphi_{o}, \text{ for all $m \geqslant 1$}.
\end{equation}
Assuming that each derivation $\partial_{m}$ equals $\frac{\partial}{\partial t_{m}}$, one arrives for (\ref{tsolAKNS}) at the solution 
$$
\psi_{0}=\psi_{0}(t,z)=\exp(\sum_{m=0}^{\infty}t_{m}Q_{0}z^{m}),
$$
$ \text{ where $t$ is the short hand notation for $\{t_{m} \mid m \geqslant 0\},$}$ and for (\ref{tsolSAKNS}) it leads to
$$
\varphi_{0}=\varphi_{0}(t,z)=\exp(\sum_{m=1}^{\infty}t_{m}Q_{0}z^{m}), \text{ with $t=\{t_{m} \mid m \geqslant 1\}.$}
$$
For general $\psi$ and $\varphi$, one needs in (\ref{LinAKNS}) resp. (\ref{LinSAKNS}) a left action with elements like $Q$, $B_{m}$ resp. $Z$, $C_{m}$ from ${\rm gl}_{2}(R)[z, z^{-1})$, an action of all the $\partial_{m}$ and a right action of $Q_{0}$ resp. $Q_{0}z$. This can all be realized in suitable ${\rm gl}_{2}(R)[z, z^{-1})$-modules that are perturbations of the trivial solutions $\psi_{0}$ and $\varphi_{0}$. For, consider in a AKNS-setting 
\begin{equation}
\label{M0}
\mathcal{M}_{0}=\left\{ \{g(z)\}\psi_{0}=\left\{ \sum_{i=-\infty}^{N}g_{i}z^{i} \right\}\psi_{0} \mid g(z)=\sum_{i=-\infty}^{N}g_{i}z^{i} \in {\rm gl}_{2}(R)[z, z^{-1}) \right\}
\end{equation}
and in a strict AKNS-setting
 \begin{equation}
\label{M1}
\mathcal{M}_{1}=\left\{ \{g(z)\} \varphi_{0}=\left\{ \sum_{i=-\infty}^{N}g_{i}z^{i} \right\}\varphi_{0} \mid g(z)=\sum_{i=-\infty}^{N}g_{i}z^{i} \in {\rm gl}_{2}(R)[z, z^{-1}) \right\},
\end{equation}
where the products $\{g(z)\}\psi_{0}$ and $\{g(z)\} \varphi_{0}$ should be seen as formal and both factors should be kept separate. On both $\mathcal{M}_{0}$ and $\mathcal{M}_{1}$ one can define the required actions: for each $h(z) \in {\rm gl}_{2}(R)[z, z^{-1})$ define
\[
h(z).\{g(z)\}\psi_{0}:=\{h(z)g(z)\}\psi_{0} \text{ resp. }h(z).\{g(z)\} \varphi_{0}:=\{h(z)g(z)\} \varphi_{0}.
\]
The right hand action of $Q_{0}$ resp. $Q_{0}z$ we define by
\[
\{g(z)\}\psi_{0}Q_{0}:=\{g(z)Q_{0}\}\psi_{0} \text{ resp. }\{g(z)\} \varphi_{0}Q_{0}z:=\{g(z)Q_{0}z\} \varphi_{0}
\]
and the action of each $\partial_{m}$ by
\begin{align*}
\partial_{m}(\{g(z)\}\psi_{0})=\left\{ \sum_{i=-\infty}^{N}\partial_{m}(g_{i})z^{i} +\left\{ \sum_{i=-\infty}^{N}g_{i}Q_{0}z^{i+m} \right\} \right\}\psi_{0}\\
\partial_{m}(\{g(z)\}\varphi_{0})=\left\{ \sum_{i=-\infty}^{N}\partial_{m}(g_{i})z^{i} +\left\{ \sum_{i=-\infty}^{N}g_{i}Q_{0}z^{i+m} \right\} \right\}\varphi_{0}
\end{align*}
Analogous to the terminology in the scalar case, see \cite{DJKM83}, we call the elements of $\mathcal{M}_{0}$ and $\mathcal{M}_{1}$ {\it oscillating matrices}.
Note that both $\mathcal{M}_{0}$ and $\mathcal{M}_{1}$ are free ${\rm gl}_{2}(R)[z, z^{-1})$-modules with generators $\psi_{0}$ resp. $\varphi_{0}$, because for 
each $h(z) \in {\rm gl}_{2}(R)[z, z^{-1})$ we have
\[
h(z).\psi_{0}=h(z).\{1\}\psi_{0}=\{h(z)\}\psi_{0} \text{ resp. }h(z).\varphi_{0}=h(z).\{1\}\psi_{0}=\{h(z)\}\varphi_{0}. 
\]
Hence, in order to be able to perform legally the scratching of the $\psi=\{h(z)\}\psi_{0}$ or $\varphi=\{h(z)\}\varphi_{0}$, it is enough to find oscillating matrices such that 
$h(z)$ is invertible in ${\rm gl}_{2}(R)[z, z^{-1})$. We will now introduce a collection of such elements that will occur at the construction of solutions of both hierarchies.

For $m=(m_{1},m_{2}) \in \mathbb{Z}^{2}$, let $\delta(m) \in {\rm gl}_{2}(R)[z, z^{-1})$ be given by
$$
\delta(m)=\left(
\begin{matrix}
z^{m_{1}}&0 \\
0&z^{m_{2}}
\end{matrix}
\right).
$$ 
Then $\delta(m)$ has $\delta(-m)$ as its inverse in ${\rm gl}_{2}(R)[z, z^{-1})$ and the collection $\Delta=\{\delta(m) \mid m \in \mathbb{Z}^{2}\}$ forms a group. An element $\psi \in \mathcal{M}_{0}$ is called an {\it oscillating matrix of type} $\delta(m)$ if it has the form 
$$
\psi=\{ h(z) \delta(m) \}\psi_{0}, \text{ with } h(z) \in G_{<0}.
$$
These oscillating matrices are examples of elements of $\mathcal{M}_{0}$ for which the scratching procedure is valid. Let $Q$ be a potential solution of the AKNS hierarchy
of the form (\ref{decoQ0}). If there is an oscillating function $\psi$ of type $\delta(m)$ such that for $Q$ and $\psi$ the equations (\ref{LinAKNS}) hold, then we call $\psi$ a {\it wave matrix of the AKNS hierarchy of type} $\delta(m)$. In particular, $Q$ is then a solution of the AKNS hierarchy and the first equation in (\ref{LinAKNS}), implies
$$
Q h(z) \delta(m)=h(z)Q_{0}\delta(m) \Rightarrow Q=h(z)Q_{0}h(z)^{-1},
$$
in other words the solution $Q$ is totally determined by the wave matrix. Similarly, we call an element $\varphi \in \mathcal{M}_{1}$ an {\it oscillating matrix of type} $\delta(m)$ if it has the form 
$$
\varphi=\{ h(z) \delta(m) \}\varphi_{0}, \text{ with } h(z) \in G_{\leqslant 0}.
$$
This type of oscillating matrices are examples of elements of $\mathcal{M}_{1}$ for which the scratching procedure is valid. Let $Z$ be a potential solution of the strict AKNS hierarchy
of the form (\ref{sdecoQ0}). If there is an oscillating function $\varphi$ of type $\delta(m)$ in $\mathcal{M}_{1}$ such that for $Z$ and $\varphi$ the equations (\ref{LinSAKNS}) hold, then we call $\varphi$ a {\it wave matrix of the strict AKNS hierarchy of type} $\delta(m)$. In particular, $Z$ is then a solution of the strict AKNS hierarchy and the first equation in (\ref{LinSAKNS}), implies
$$
Z h(z) \delta(m)=h(z)Q_{0}z\delta(m) \Rightarrow Z=h(z)Q_{0}zh(z)^{-1},
$$
so that also here the wave matrix fully determines the solution. 

For both hierarchies, there is a milder condition that oscillating matrices of a certain type have to satisfy, in order to become a wave function of that hierarchy.
\begin{proposition}
\label{P2.1}
Let $\psi=\{ h(z) \delta(m) \}\psi_{0}$ be an oscillating matrix of type $\delta(m)$ in $\mathcal{M}_{0}$ and $Q=h(z)Q_{0}h(z)^{-1}$ the corresponding potential solution of the AKNS hierarchy. Similarly, let $\varphi=\{ h(z) \delta(m) \}\varphi_{0}$ be such a matrix in $\mathcal{M}_{1}$ with potential solution $Z=h(z)Q_{0}zh(z)^{-1}$.
\begin{enumerate}
\item[(a)] If there exists for each $m \geqslant 0$ an element $M_{m} \in {\rm gl}_{2}(R)[z, z^{-1})_{\geqslant 0}$ such that 
\[
\partial_{m}(\psi)=M_{m}\psi.
\]
Then each $M_{m}=(Qz^{m})_{\geqslant 0}$ and $\psi$ is a wave matrix of type $\delta(m)$ for the AKNS hierarchy.
\item[(b)] 
If there exists for each $m \geqslant 1$ an element $N_{m} \in {\rm gl}_{2}(R)[z, z^{-1})_{> 0}$ such that 
\[
\partial_{m}(\varphi)=N_{m}\varphi.
\]
Then each $N_{m}=(Zz^{m-1})_{>0}$ and $\varphi$ is a wave matrix of type $\delta(m)$ for the strict AKNS hierarchy.
\end{enumerate}
\end{proposition}
\begin{proof} We give the proof again for the strict AKNS-case, the other one is similar. By using the fact that $\mathcal{M}_{1}$ is a free ${\rm gl}_{2}(R)[z, z^{-1})$-module 
with generator $\varphi_{0}$ we can translate the relations $\partial_{m}(\varphi)=N_{m}\varphi$ into equations in ${\rm gl}_{2}(R)[z, z^{-1})$. This yields:
$$
\partial_{m}(h(z))+h(z)Q_{0}z^{m}=N_{m}h(z) \Rightarrow \partial_{m}(h(z))h(z)^{-1} +Zz^{m-1}=N_{m}
$$
Projecting the right hand side on ${\rm gl}_{2}(R)[z, z^{-1})_{>0}$ gives us the identity $$(Zz^{m-1})_{>0}=N_{m}$$ we are looking for.
\end{proof}
In the next section we produce a context where we can construct wave matrices for both hierarchies in which the product is real.

\section{The construction of the solutions}
\label{constr}

In this section we will show how to construct a wide class of solutions of both hierarchies. This is done in the style of \cite{SW}. We first describe the group of loops we will work with. For each $0 < r<1$, let $A_{r}$ be the annulus $$\{ z\mid z \in \mathbb{C}, r \leqslant |z| \leqslant \dfrac{1}{r} \}.$$  Following \cite{PS}, we use the notation $L_{an}\GL_{2}(\mathbb{C})$ for the collection of holomorphic maps from some annulus $A_{r}$ into $\GL_{2}(\mathbb{C})$. It is a group w.r.t. point wise multiplication and contains in a natural way $\GL_{2}(\mathbb{C})$ as a subgroup as the collection of constant maps into $\GL_{2}(\mathbb{C})$. Other examples of elements in $L_{an}\GL_{2}(\mathbb{C})$  are the elements of $\Delta$. However, $L_{an}\GL_{2}(\mathbb{C})$ is more than just a group, it is an infinite dimensional Lie group.
Its manifold structure can be read off from its Lie algebra $L_{an}\gl_{2}(\mathbb{C})$ consisting of all holomorphic maps $\gamma :U \to \gl_{2}(\mathbb{C})$, where $U$ is an open neighborhood of some annulus $A_{r}, 0 < r<1.$ Since $\gl_{2}(\mathbb{C})$ is a Lie algebra, the space $L_{an}\gl_{2}(\mathbb{C})$ becomes a Lie algebra w.r.t. the point wise commutator. Topologically, the space $L_{an}\gl_{2}(\mathbb{C})$ is the direct limit of all the spaces $L_{an,r}\gl_{2}(\mathbb{C})$, where this last space consists of all $\gamma$ corresponding to the fixed annulus $A_{r}$. One gives each $L_{an,r}\gl_{2}(\mathbb{C})$ the topology of uniform convergence and with that topology it becomes a Banach space. In this way, $L_{an}\gl_{2}(\mathbb{C})$ becomes a Fr\'echet space. The point wise exponential map defines a local diffeomorphism around zero in $L_{an}\gl_{2}(\mathbb{C})$, see e.g. \cite{Hamilton82}.

Now each $\gamma \in L_{an}\gl_{2}(\mathbb{C})$ possesses an expansion in a Fourier series
\begin{equation}
\label{FSgamma}
\gamma= \sum_{k=-\infty}^{\infty}\gamma_{k}z^{k}, \text{ with each }\gamma_{k} \in \gl_{2}(\mathbb{C})
\end{equation}
that converges absolutely on the annulus it is defined:
\[
\sum_{k=-\infty}^{\infty} ||\gamma_{k}||r^{-|k|} < \infty.
\]
This Fourier expansion is used to make two different decompositions of the Lie algebra $L_{an}\gl_{2}(\mathbb{C})$. The first is the analogue for the present Lie algebra $L_{an}\gl_{2}(\mathbb{C})$ of the decomposition (\ref{AKNSdeco}) of ${\rm gl}_{2}(R)[z, z^{-1})$ that lies at the basis of the Lax equations of the AKNS hierarchy. Namely, consider the subspaces 
\begin{align*}
L_{an}\gl_{2}(\mathbb{C})_{\geqslant 0}:=\{ \gamma \mid \gamma \in L_{an}\gl_{2}(\mathbb{C}),\gamma =\sum_{k=0}^{\infty}\gamma_{k}z^{k}\}\\
L_{an}\gl_{2}(\mathbb{C})_{< 0}:=\{ \gamma \mid \gamma \in L_{an}\gl_{2}(\mathbb{C}),\gamma =\sum_{k=-\infty}^{-1}\gamma_{k}z^{k}\}
\end{align*}
Both are Lie subalgebras of $L_{an}\gl_{2}(\mathbb{C})$ and their direct sum equals the whole Lie algebra. The first Lie algebra consists of the elements in $L_{an}\gl_{2}(\mathbb{C})$ that extend to holomorphic maps defined on some disk $\{ z \in \mathbb{C} \mid |z| \leqslant \dfrac{1}{r} \},0< r <1,$ and the second Lie algebra corresponds to the maps in $L_{an}\gl_{2}(\mathbb{C})$ that have a holomorphic extension towards some disk around infinity $\{z \in \mathbb{P}^{1}(\mathbb{C} ) \mid |z| \geqslant r \},0< r <1,$ and that are zero at infinity.
To each of the two Lie subalgebras belongs a subgroup of $L_{an}\GL_{2}(\mathbb{C})$. The point wise exponential map applied to elements of $L_{an}\gl_{2}(\mathbb{C})_{< 0}$ yields elements of 
\[
U_{-}=\{ \gamma \mid \gamma \in L_{an}\gl_{2}(\mathbb{C}),\gamma =\Id+\sum_{k=-\infty}^{-1}\gamma_{k}z^{k}\}
\]
and the exponential map applied to elements of $L_{an}\gl_{2}(\mathbb{C})_{\geqslant 0}$ maps them into
\[
P_{+}=\{ \gamma \mid \gamma \in L_{an}\gl_{2}(\mathbb{C}),\gamma =\gamma_{0}+\sum_{k=1}^{\infty}\gamma_{k}z^{k}, \text{ with }\gamma_{0} \in \GL_{2}(\mathbb{C})\}.
\]
As well $U_{-}$ as $P_{+}$ are easily seen to be subgroups of $L_{an}\GL_{2}(\mathbb{C})$ and since the direct sum of their Lie algebras is $L_{an}\gl_{2}(\mathbb{C})$, their product 
\begin{equation}
\label{U-P+}
\Omega=U_{-}P_{+}
\end{equation}
is open in $L_{an}\GL_{2}(\mathbb{C})$ and is called, like in the finite dimensional case, the {\it big cell} w.r.t. $U_{-}$ and $P_{+}$.

The second decomposition is the analogue of the splitting (\ref{SAKNSdeco}) of ${\rm gl}_{2}(R)[z, z^{-1})$ that led to the Lax equations of the strict AKNS hierarchy.
Now we consider the subspaces 
\begin{align*}
L_{an}\gl_{2}(\mathbb{C})_{> 0}:=\{ \gamma \mid \gamma \in L_{an}\gl_{2}(\mathbb{C}),\gamma =\sum_{k=1}^{\infty}\gamma_{k}z^{k}\}\\
L_{an}\gl_{2}(\mathbb{C})_{\leqslant 0}:=\{ \gamma \mid \gamma \in L_{an}\gl_{2}(\mathbb{C}),\gamma =\sum_{k=-\infty}^{0}\gamma_{k}z^{k}\}
\end{align*}
Both are Lie subalgebras of $L_{an}\gl_{2}(\mathbb{C})$ and their direct sum equals the whole Lie algebra. $L_{an}\gl_{2}(\mathbb{C})_{> 0}$ consists of maps whose holomorphic extension to $z=0$ equals zero in that point and $L_{an}\gl_{2}(\mathbb{C})_{\leqslant 0}$ is the set of maps that extend homomorphically to a neighborhood of infinity. To each of them belongs a subgroup of $L_{an}\GL_{2}(\mathbb{C})$. The point wise exponential map applied to elements of $L_{an}\gl_{2}(\mathbb{C})_{> 0}$ yields elements of 
\[
U_{+}=\{ \gamma \mid \gamma \in L_{an}\gl_{2}(\mathbb{C}),\gamma =\Id+\sum_{k=1}^{\infty}\gamma_{k}z^{k}\}
\]
and the exponential map applied to elements of $L_{an}\gl_{2}(\mathbb{C})_{\leqslant 0}$ maps them into
\[
P_{-}=\{ \gamma \mid \gamma \in L_{an}\gl_{2}(\mathbb{C}),\gamma =\gamma_{0}+\sum_{k=-\infty}^{-1}\gamma_{k}z^{k}, \text{ with }\gamma_{0} \in \GL_{2}(\mathbb{C})\}.
\]
As well $U_{-}$ as $P_{+}$ are easily seen to be subgroups of $L_{an}\GL_{2}(\mathbb{C})$. Since the direct sum of their Lie algebras is $L_{an}\gl_{2}(\mathbb{C})$, their product $P_{-}U_{+}$ is open and because $P_{+}=\GL_{2}(\mathbb{C})U_{+}=U_{+}\GL_{2}(\mathbb{C})$ and $P_{-}=\GL_{2}(\mathbb{C})U_{-}=U_{-}\GL_{2}(\mathbb{C})$, this gives the equality
\begin{equation}
\label{P-U+}
P_{-}U_{+}=\Omega=U_{-}P_{+},
\end{equation}
the big cell in $L_{an}\GL_{2}(\mathbb{C})$.

The next two subgroups of $L_{an}\GL_{2}(\mathbb{C})$ correspond to the exponential factors in both linearizations. The group of commuting flows relevant for the AKNS hierarchy is the group
\[
\Gamma_{0}=\{ \gamma_{0}(t)=\exp(\sum_{i=0}^{\infty} t_{i}Q_{0}z^{i}) \mid \gamma_{0} \in P_{+} \}
\]
and similarly, for the strict AKNS hierarchy, we use the commuting flows from
\[
\Gamma_{1}=\{ \gamma_{1}(t)=\exp(\sum_{i=1}^{\infty} t_{i}Q_{0}z^{i}) \mid \gamma_{1} \in U_{+} \}.
\]
Besides the groups $\Delta$, $\Gamma_{0}$ and $\Gamma_{1}$, there are more subgroups in $L_{an}\GL_{2}(\mathbb{C})$ that commute with those three groups and they are in a sense complimentary to $\Delta$, $\Gamma_{0}$ resp. $\Delta$, $\Gamma_{1}$. That is why we use the following notations for them: in $U_{-}$ there is 
$$
\Gamma_{0}^{c}=\{ \gamma^{c}_{0}(t)=\exp(\sum_{k=1}^{\infty} s_{k}Q_{0}z^{-k}) \mid \gamma^{c}_{0} \in U_{-} \}
$$
and in $P_{-}$ we have 
$$
\Gamma^{c}_{1}=\{ \gamma^{c}_{1}(t)=\exp(\sum_{k=0}^{\infty} s_{k}Q_{0}z^{-k}) \mid \gamma^{c}_{1} \in P_{-} \}.
$$

We have now all ingredients to describe the construction of the solutions to each hierarchy and we start with the AKNS hierarchy. Take inside the product $L_{an}\GL_{2}(\mathbb{C}) \times \Delta$ the collection $S_{0}$ of pairs $(g,\delta(m))$ such that there exists a $\gamma_{0}(t) \in \Gamma_{0}$ satisfying 
\begin{equation}
\label{CC0}
\delta(m) \gamma_{0}(t) g \gamma_{0}(t)^{-1}\delta(-m) \in \Omega=U_{-}P_{+}
\end{equation}
For such a pair $(g,\delta(m))$ we take the collection $\Gamma_{0}(g,\delta(m))$ of all $\gamma_{0}(t)$ satisfying the condition (\ref{CC0}). This is an open non-empty subset of $\Gamma_{0}$. Let $R_{0}(g,\delta(m))$ be the algebra of analytic functions $\Gamma_{0}(g,\delta(m)) \to \mathbb{C}$. This is the algebra of functions $R$ that we associate with the point $(g,\delta(m)) \in S_{0}$ and for the commuting derivations of $R_{0}(g,\delta(m))$ we choose $\partial_{i}:=\dfrac{\partial}{\partial t_{i}}, i \geqslant 0$. By property (\ref{CC0}), we have for all $\gamma_{0}(t) \in \Gamma_{0}(g,\delta(m))$
\begin{equation}
\label{decoU-P+}
\delta(m) \gamma_{0}(t) g \gamma_{0}(t)^{-1}\delta(-m)=u_{-}(g,\delta(m))^{-1}p_{+}(g,\delta(m)), 
\end{equation}
with $u_{-}(g,\delta(m)) \in U_{-}, p_{+}(g,\delta(m)) \in P_{+}.$ Then all the matrix coefficients in the Fourier expansions of the elements $u_{-}(g,\delta(m))$ and $p_{+}(g,\delta(m))$ belong to the algebra $R_{0}(g,\delta(m))$. It is convenient to write the relation (\ref{decoU-P+}) in the form
\begin{align}
\label{wmAKNS}
\Psi_{g,\delta(m)}:&=u_{-}(g,\delta(m))\delta(m) \gamma_{0}(t)=p_{+}(g,\delta(m))\delta(m)\gamma_{0}(t)g^{-1}\\ \notag
&=p_{+}(g,\delta(m))\gamma_{0}(t)\delta(m)g^{-1}=q_{+}(g,\delta(m))\delta(m)g^{-1}, 
\end{align}
with $q_{+}(g,\delta(m)) \in P_{+}.$ Clearly, $\Psi_{g,\delta(m)}$ is an oscillating matrix of type $\delta(m)$ in $\mathcal{M}_{0}$, where all the products in the decomposition $u_{-}(g,\delta(m))\delta(m) \gamma_{0}(t)$ are no longer formal but real. This is our potential wave matrix for the AKNS hierarchy.

To get the potential wave function for the strict AKNS hierarchy, we proceed similarly, only this time we use the decomposition $\Omega=P_{-}U_{+}$.
Consider again the product $L_{an}\GL_{2}(\mathbb{C}) \times \Delta$ and the subset $S_{1}$ of pairs $(g,\delta(m))$ such that there exists a $\gamma_{1}(t) \in \Gamma_{1}$ satisfying 
\begin{equation}
\label{CC1}
\delta(m) \gamma_{1}(t) g \gamma_{1}(t)^{-1}\delta(-m) \in \Omega=P_{-}U_{+}.
\end{equation}
For such a pair $(g,\delta(m))$ we define $\Gamma_{1}(g,\delta(m))$ as the set of all $\gamma_{1}(t)$ satisfying the condition (\ref{CC0}). This is an open non-empty subset of $\Gamma_{1}$. Let $R_{1}(g,\delta(m))$ be the algebra of analytic functions $\Gamma_{1}(g,\delta(m)) \to \mathbb{C}$. This is the algebra of functions $R$ that we associate with the point $(g,\delta(m)) \in S_{1}$ and for the commuting derivations of $R_{1}(g,\delta(m))$ we choose $\partial_{i}:=\dfrac{\partial}{\partial t_{i}}, i >0$. By property (\ref{CC1}), we have for all $\gamma_{1}(t) \in \Gamma_{1}(g,\delta(m))$
\begin{equation}
\label{decoP-U+}
\delta(m) \gamma_{1}(t) g \gamma_{1}(t)^{-1}\delta(-m)=p_{-}(g,\delta(m))^{-1}u_{+}(g,\delta(m)), 
\end{equation}
with $p_{-}(g,\delta(m)) \in P_{-}, u_{+}(g,\delta(m)) \in U_{+}.$ Then all the matrix coefficients in the Fourier expansions of the elements $p_{-}(g,\delta(m))$ and $u_{+}(g,\delta(m))$ belong to the algebra $R_{1}(g,\delta(m))$. It is convenient to write the relation (\ref{decoP-U+}) in the form
\begin{align}
\label{wmSAKNS}
\Phi_{g,\delta(m)}:&=p_{-}(g,\delta(m))\delta(m) \gamma_{1}(t)=u_{+}(g,\delta(m))\delta(m)\gamma_{1}(t)g^{-1}\\ \notag
&=u_{+}(g,\delta(m))\gamma_{1}(t)\delta(m)g^{-1}=w_{+}(g,\delta(m))\delta(m)g^{-1}, 
\end{align}
with $w_{+}(g,\delta(m)) \in U_{+}$. Clearly, $\Phi_{g,\delta(m)}$ is an oscillating matrix of type $\delta(m)$ in $\mathcal{M}_{1}$ where all the products in the decomposition $p_{-}(g,\delta(m))\delta(m) \gamma_{1}(t)$ are no longer formal but real.

To show that $\Psi_{g,\delta(m)}$ is a wave matrix for the AKNS hierarchy of type $\delta(m)$ and $\Phi_{g,\delta(m)}$ one for the strict AKNS hierarchy, we use the fact that it suffices to prove the relations in Proposition \ref{P2.1}. We treat first $\Psi_{g,\delta(m)}$. In equation (\ref{wmAKNS}) we presented two different decompositions of $\Psi_{g,\delta(m)}$ that we will use to compute $\partial_{i}(\Psi_{g,\delta(m)})$. W.r.t. $u_{-}(g,\delta(m))\delta(m) \gamma_{0}(t)$ there holds 
\[ 
\partial_{i}(\Psi_{g,\delta(m)})=\{ \partial_{i}(u_{-}(g,\delta(m)))u_{-}(g,\delta(m))^{-1}+z^{i}\}u_{-}(g,\delta(m))\delta(m) \gamma_{0}(t)=M_{i}\Psi_{g,\delta(m)},
\]
with $M_{i} \in {\rm gl}_{2}(R(g,\delta(m)))[z, z^{-1})$ of order $i$ in $z$. If we take the decomposition $q_{+}(g,\delta(m))\delta(m)g^{-1}$, then we get 
\[
\partial_{i}(\Psi_{g,\delta(m)})=\{ \partial_{i}(q_{+}(g,\delta(m)))q_{+}(g,\delta(m))^{-1}\}\Psi_{g,\delta(m)}
\]
with $\partial_{i}(q_{+}(g,\delta(m)))q_{+}(g,\delta(m))^{-1}$ of  the form $\sum_{j \geqslant 0}q_{j}z^{j}$, with all $q_{j} \in \gl_{2}(R(g,\delta(m)))$. As $\Psi_{g,\delta(m)}$ is invertible, we obtain $M_{i}=\sum_{j \geqslant 0}q_{j}z^{j}$ and thus $M_{i}=\sum_{j = 0}^{i}q_{j}z^{j}$ and that is the desired result.

We handle the case of $\Phi_{g,\delta(m)}$ in a similar way. Also in equation (\ref{wmAKNS}) you can find two different decompositions of $\Phi_{g,\delta(m)}$ that we will use to compute $\partial_{i}(\Phi_{g,\delta(m)})$. W.r.t. the expression $p_{-}(g,\delta(m))\delta(m) \gamma_{1}(t)$ we get for all $i \geqslant 1$
\[ 
\partial_{i}(\Phi_{g,\delta(m)})=\{ \partial_{i}(p_{-}(g,\delta(m)))p_{-}(g,\delta(m))^{-1}+z^{i}\}p_{-}(g,\delta(m))\delta(m) \gamma_{1}(t)=N_{i}\Phi_{g,\delta(m)},
\]
again with $N_{i} \in {\rm gl}_{2}(R_{1}(g,\delta(m)))[z, z^{-1})$ of order $i$ in $z$. Next we take the decomposition $w_{+}(g,\delta(m))\delta(m)g^{-1}$ and that yields
\[
\partial_{i}(\Phi_{g,\delta(m)})=\{ \partial_{i}(w_{+}(g,\delta(m)))w_{+}(g,\delta(m))^{-1}\}\Psi_{g,\delta(m)}.
\]
Since the zero-th order term of $w_{+}(g,\delta(m))$ is $\Id$, we note that the element\\ $\partial_{i}(w_{+}(g,\delta(m)))w_{+}(g,\delta(m))^{-1}$ has the form $\sum_{j \geqslant 1}w_{j}z^{j}$, with all $w_{j} \in \gl_{2}(R_{1}(g,\delta(m)))$. As $\Phi_{g,\delta(m)}$ is also invertible, we obtain $N_{i}=\sum_{j \geqslant 1}w_{j}z^{j}$ and thus $N_{i}=\sum_{j = 1}^{i}w_{j}z^{j}$ and that is the result we were looking for in the strict case.

The above constructions are not affected seriously by a number of transformations. E.g. if we write $\delta=\delta((1,1))$, then all the $\delta^{k}, k \in \mathbb{Z},$ are central elements and we have for all $(g,\delta(m)) \in S_{0} \text{ resp. } S_{1}$ and all $k \in \mathbb{Z}$ that 
\[
\Psi_{g,\delta(m)\delta^{k}}=\Psi_{g,\delta(m)}\delta^{k} \text{ resp. }\Phi_{g,\delta(m)\delta^{k}}=\Phi_{g,\delta(m)}\delta^{k}.
\]
Also for any $p_{1} \in P_{+}$, each $\gamma^{c}_{0}(s) \in \Gamma^{c}_{0}$ and all $(g,\delta(m)) \in S_{0}$, we see that the element $(\gamma^{c}_{0}(s)g\delta(-m)p_{1}\delta(m),\delta(m))$ also belongs to $S_{0}$, for 
\[
\delta(m)\gamma_{0}(t)\gamma^{c}_{0}(s)g\delta(-m)p_{1}\delta(m)\gamma_{0}(t)^{-1}\delta(-m)=\gamma^{c}_{0}(s)u_{-}^{-1}p_{+}\gamma_{0}(t)p_{1}\gamma_{0}(t)^{-1}.
\]
In particular, we see that 
$$
\Psi_{\gamma^{c}_{0}(s)g\delta(-m)p_{1}\delta(m),\delta(m)}=u_{-}\gamma^{c}_{0}(s)^{-1}\delta(m)\gamma_{0}(t).
$$ 
Its analogue in the strict case is: for any $u_{1} \in U_{+}$, each $\gamma^{c}_{1}(s) \in \Gamma^{c}_{1}$ and all $(g,\delta(m)) \in S_{1}$ the element $(\gamma^{c}_{1}(s)g\delta(-m)u_{1}\delta(m),\delta(m))$ also belongs to $S_{1}$ and there holds:
\[
\Phi_{\gamma^{c}_{1}(s)g\delta(-m)u_{1}\delta(m),\delta(m)}=p_{-}\gamma^{c}_{1}(s)^{-1}\delta(m)\gamma_{1}(t)
\]
We resume the foregoing results in a

\begin{theorem}
\label{T4.1}
Consider the product space $\Pi:=L_{an}\GL_{2}(\mathbb{C}) \times \Delta$. Then there holds:
\begin{enumerate}
\item[$($a$)$] For each point $(g,\delta(m)) \in \Pi$ that satisfies the condition (\ref{CC0}), we have a wave matrix of type $\delta(m)$ of the AKNS hierarchy $\Psi_{g,\delta(m)}$,
defined by (\ref{wmAKNS}), and this determines a solution $Q_{g,\delta(m)}=u_{-}(g,\delta(m))Q_{0}u_{-}(g,\delta(m))^{-1}$ of the AKNS hierarchy. For each $p_{1}\in P_{+}$, every $\gamma^{c}_{0}(s) \in \Gamma^{c}_{0}$
and all powers of $\delta$ one has:
$$
Q_{\gamma^{c}_{0}(s)g\delta(-m)p_{1}\delta(m), \delta^{k}\delta(m)}=Q_{g,\delta(m)}.
$$
\item[$($b$)$]
For each point $(g,\delta(m)) \in \Pi$ that satisfies the condition (\ref{CC1}), we have a wave matrix of type $\delta(m)$ of the strict AKNS hierarchy $\Phi_{g,\delta(m)}$,
defined by (\ref{wmSAKNS}), and it determines a solution $Z_{g,\delta(m)}=p_{-}(g,\delta(m))Q_{0}z p_{-}(g,\delta(m))^{-1}$ of the strict AKNS hierarchy. For each $u_{1}\in U_{+}$, every $\gamma^{c}_{1}(s) \in \Gamma^{c}_{1}$
and all powers of $\delta$ one has:
$$
Z_{\gamma^{c}_{1}(s)g\delta(-m)u_{1}\delta(m), \delta^{k}\delta(m)}=Z_{g,\delta(m)}.
$$
\end{enumerate}
\end{theorem}


\begin{thebibliography}{10}


\bibitem{A-vM-VH} M. Adler, P. van Moerbeke, P. Vanhaecke, {\em Algebraic Integrability, Painlev\'e Geometry and Lie Algebras}, Ergebnisse der Mathematik und ihrer Grenzgebiete, volume 47, Springer Verlag.

\bibitem{AKNS} M.~J. Ablowitz, D.~J.Kaup, A.~C. Newell, H. Segur,  {\em The inverse scattering transform-Fourier analysis for nonlinear problems},
Studies in Appl. Math. 53 (4) (1974): 249-315.


\bibitem{DJKM83} E. Date, M. Jimbo, M. Kashiwara, T. Miwa,
{\em Transformation groups for soliton equations}, Proceedings of the RIMS symposium on nonlinear integrable systems-Classical Theory and Quantum Theory, ed. M.Jimbo and T.Miwa, World Sci.Publishers, Singapore.


\bibitem{FNR} H. Flashka, A.~C.Newell, T. Ratiu, \emph{Kac-Moody Lie Algebras and soliton equations II; Lax equations associated with $A_1^{(1)}$}, Physica 9D, p.300-323.


\bibitem{Hamilton82} R.~S. Hamilton, {\em The inverse function theorem of Nash and Moser.} Bull. Amer. Math. Soc. (N.S.) 7 (1982), no. 1, 65-222.


\bibitem{HHP} G.~F. Helminck, A.~G. Helminck, E.~A. Panasenko,  {\em Integrable deformations in the algebra of pseudo differential operators from a Lie algebraic perspective}, Theoretical and Mathematical Physics {\bf 174}(1): 134-153 (2013).


\bibitem{PS} A. Pressley, G.Segal, {\em Loop groups}, Oxford Mathematical Monographs, Clarendon Press 1986.

\bibitem{SW} G. Segal and G. Wilson, {\em Loop groups and equations of
KdV type}, Publ. Math. IHES {\bf 63} (1985), 1--64.


\end{thebibliography}
\end{document}